\documentclass[12pt,a4paper,fleqn]{article}

\usepackage[margin=1in]{geometry}
\usepackage[T1]{fontenc}
\usepackage{amsmath,dsfont}
\usepackage{caption}
\usepackage[utf8]{inputenc}
\usepackage{bm}
\usepackage{verbatim}
\usepackage{float}
\usepackage{mwe}
\usepackage{color,soul}
\usepackage{threeparttable}
\usepackage{graphicx}
\usepackage{subcaption}
\usepackage{mwe}
\usepackage{mdframed}
\usepackage{xcolor}
\usepackage{natbib}
\usepackage{multirow}
\usepackage{makecell}
\usepackage{booktabs}
\usepackage{array}
\usepackage{hhline}

\usepackage{lscape}
\usepackage{mathtools}
\usepackage{bbm,amsfonts}
\usepackage{romanbar}
\usepackage{amsthm}
\usepackage{enumerate}
\usepackage{rotating}
\usepackage{tikz}
\usetikzlibrary{matrix}
\usetikzlibrary{calc,intersections}
\usepackage{mathrsfs}

\usepackage{color}

\newtheorem{theorem}{Theorem}[section]
\newtheorem{corollary}[theorem]{Corollary}
\newtheorem{lemma}[theorem]{Lemma}

\theoremstyle{definition}

\theoremstyle{remark}

\numberwithin{equation}{section}

\newcommand{\bol}[1]{\mbox{\boldmath$#1$}}

\newcommand{\bSigma}{\mathbf{\Sigma}}

\newcommand{\bmu}{\bol{\mu}}

\newcommand{\btheta}{\bol{\theta}}
\newcommand{\bb}{\mathbf{b}}

\newcommand{\bx}{\mathbf{x}}

\newcommand{\by}{\mathbf{y}}

\newcommand{\bC}{\mathbf{C}}

\newcommand{\bX}{\mathbf{X}}

\newcommand{\bw}{\mathbf{w}}
\newcommand{\hbw}{\mathbf{\hat{w}}}

\newcommand{\bOne}{\mathbf{1}}

\newcommand{\bI}{\mathbf{I}}

\newcommand{\bxi}{\boldsymbol{\xi}}
\newcommand{\bD}{\mathbf{D}}

\newcommand{\bV}{\mathbf{V}}

\newcommand{\bS}{\mathbf{S}}

\providecommand{\keywords}[1]
{
\small	
\textbf{\textit{Keywords}:} #1
}

\usepackage{authblk}
\title{Is the empirical out-of-sample variance an informative risk measure for the high-dimensional portfolios?}
\author[1]{Taras Bodnar}
\author[2]{Nestor Parolya}
\author[1]{Erik Thorsén}
\affil[1]{Department of Mathematics, Stockholm University, Roslagsv\"{a}gen 101, SE-10691 Stockholm, Sweden}
\affil[2]{Department of Applied Mathematics, Delft University of Technology, Mekelweg 4,
2628 CD Delft, The Netherlands}
\date{\today}

\begin{document}

% This footnote should be put closed to the author who will submit the paper.
%\footnote{Corresponding author: . E-mail address: .}

\maketitle
% 2021-09-06 1013671

\begin{abstract}
The main contribution of this paper is the derivation of the asymptotic behaviour of the out-of-sample variance, the out-of-sample relative loss, and of their empirical counterparts in the high-dimensional setting, i.e., when both ratios $p/n$ and $p/m$ tend to some positive constants as $m\to\infty$ and $n\to\infty$, where $p$ is the portfolio dimension, while $n$ and $m$ are the sample sizes from the in-sample and out-of-sample periods, respectively. The results are obtained for the traditional estimator of the global minimum variance (GMV) portfolio, for the two shrinkage estimators introduced by \cite{frahm2010} and \cite{bodnar2018estimation}, and for the equally-weighted portfolio, which is used as a target portfolio in the specification of the two considered shrinkage estimators. We show that the behaviour of the empirical out-of-sample variance may be misleading is many practical situations. On the other hand, this will never happen with the empirical out-of-sample relative loss, which seems to provide a natural normalization of the out-of-sample variance in the high-dimensional setup. As a result, an important question arises if this risk measure can safely be used in practice for portfolios constructed from a large asset universe.
\end{abstract}


\newpage
\section{Introduction}\label{sec:intro}
Mean-variance analysis of Markowitz is a well established tool for optimal portfolio selection which is one of the most popular approaches today in financial literature (see, e.g., \cite{markowitz1952}, \cite{markowitz1959portfolio}, \cite{britten1999sampling}, \cite{ao2019approaching}, \cite{bodnar2020tests},  \cite{ding2021high}). The idea behind the approach is to invest in the portfolio which has the smallest variance for a given level of the expected return. In the limiting case of the fully risk-averse investor, the so-called global minimum variance (GMV) portfolio is selected. The latter portfolio possesses the smallest variance among all mean-variance optimal portfolios and lies on the vertex of the efficient frontier which is a parabola in the mean-variance space (see, \cite{merton1972}, \cite{kan2008distribution}, \cite{Bodnar2009Econometricalanalysisofthesampleefficientfrontier}).

One of the important challenges, which arise when the Markowitz theory is implemented in practice, is related to the estimation error which appears when unknown parameters of the data-generating process are replaced by their sample counterparts in the expressions of the optimal portfolio weights (see, \cite{okhrin2006distributional}, \cite{el2010high}, \cite{cai2020high}, \cite{bodnarokhrinparolya2020}, \cite{bodnar2021recent}). The impact of the parameter uncertainty on the performance of optimal portfolios is usually comparable to or even larger than the one described by the model uncertainty which is determined by using the covariance matrix in Markowitz optimization problem. Moreover, the estimation error present in an estimator of the mean vector has even a larger influence on the performance of optimal portfolios than the error related to the estimation of the covariance matrix (see, e.g., \cite{merton1980estimating}, \cite{best1991sensitivity}, \cite{chopra1993effect}). This is usually used in financial literature as an argument to hold the GMV portfolio whose weights only depends on the covariance matrix(\cite{chan1999portfolio}, \cite{jagannathan2003risk}, \cite{frahm2010}, \cite{bodnar2021sampling}).

Let $\by$ denote the $k$-dimensional vector of the asset returns and let $\bmu=\mathbb{E}(\by)$ and $\bSigma=\mathbb{V}ar(\by)$ be its mean vector and covariance matrix. Then the expected return and the variance of the portfolio with the weights $\bw$ are given by
\begin{equation*}
R_p=\bw^\top \bmu
\quad \text{and} \quad
V_p=\bw^\top\bSigma\bw,
\end{equation*}
respectively. The weights of the GMV portfolio are found by minimizing $V_p$ given that the whole investor wealth is invested in the selected assets, i.e., under the constraint $\bw^\top \bOne=1$ where $\bOne$ denotes the $p$-dimensional vector of ones. They are given by
\begin{equation}\label{w_gmv}
\bw_{GMV}=\frac{\bSigma^{-1}\bOne}{\bOne^\top\bSigma^{-1}\bOne},
\end{equation}
while the variance of the GMV portfolio is expressed as
\begin{equation}\label{V_gmv}
V_{GMV}=\bw_{GMV}^\top \bSigma \bw_{GMV}=\frac{1}{\bOne^\top\bSigma^{-1}\bOne}.
\end{equation}
We refer to $\bw_{GMV}$ and $V_{GMV}$ as the population weights and the population variance of the GMV portfolio, since they both depend on the unknown parameter $\bSigma$ of the data-generating model. It has to be noted that $V_{GMV}$ is also called the in-sample variance in financial literature (see, \cite{frahm2010}).

In practical applications, the population GMV portfolio cannot be constructed since its weights $\bw_{GMV}$ depend on the unobservable quantity $\bSigma$. Given historical realizations of the asset returns, $\by_1,...,\by_n$, the population covariance matrix is estimated by its sample counterpart expressed as
\begin{equation}\label{bS}
\bS_n=\frac{1}{n-1} \sum_{i=1}^n (\by_i-\bar{\by}_n)(\by_i-\bar{\by}_n)^\top
\quad \text{with} \quad
\bar{\by}_n=\frac{1}{n} \sum_{i=1}^n \by_i.
\end{equation}
Then, the traditional GMV portfolio is determined as the sample estimator of $\bw_{GMV}$ where the unknown $\bSigma$ is replaced by $\bS_n$, i.e.,
\begin{equation}\label{w_S}
\hbw_{n;S}=\frac{\bS_n^{-1}\bOne}{\bOne^\top\bS_n^{-1}\bOne}.
\end{equation}
If the portfolio dimension $p$ is considerably smaller than the sample size $n$, then $\bS_n$ consistently estimates $\bSigma$ under weak conditions imposed on the data-generating model of the asset returns and, consequently, the traditional GMV portfolio provides a good approximation of the population GMV portfolio.

The situation is completely different in the high-dimensional setting when the portfolio dimension is comparable to the sample size such that $p/n \to c \in [0,1)$ as $n \to \infty$ where the constant $c$ is called the concentration ratio (see, \cite{bai2010spectral}, \cite{bodnardetteparolya2019}). In this case the sample covariance matrix $\bS_n$ is not longer a consistent estimator for $\bSigma$. As a result, the traditional GMV portfolio might deviate considerably from the population GMV portfolio. In order to ensure a good performance of the holding portfolio, the weights of the traditional GMV portfolio have to be adjusted by taking the parameter uncertainty into account (see, e.g., \cite{jagannathan2003risk}, \cite{BodnarDmytriv2019}, \cite{ao2019approaching}, \cite{cai2020high}, \cite{ding2021high}).

In order to define an improved estimator of the high-dimensional GMV portfolio, i.e., when $p$ is comparable to $n$, the optimization problem has to be formulated. As a performance measure, the out-of-sample variance is usually used which is given by
\begin{equation}\label{oosVar}
V_{\hat{\bw}_n}= \hat{\bw}_n^\top\bSigma \hat{\bw}_n,   
\end{equation}
where $\hat{\bw}_n$ is an estimator of $\bw_{GMV}$ based on the asset returns $\by_1,...,\by_n$. Alternatively, one can use the out-of-sample relative loss
\begin{equation}\label{oosLoss}
L_{\hat{\bw}_n}=\frac{V_{\hat{\bw}_n}-V_{GMV}}{V_{GMV}}= \bOne^\top\bSigma^{-1}\bOne\hat{\bw}_n^\top \bSigma \hat{\bw}_n-1,   
\end{equation}
as a performance measure. By definitions of $V_{\hat{\bw}_n}$ and $L_{\hat{\bw}_n}$, one directly gets that the portfolio which minimizes the out-of-sample variance also minimizes the out-of-sample relative loss and vice versa.

Unfortunately, due to the presence of $\bSigma$ in \eqref{oosVar} and in \eqref{oosLoss}, both the performance measures can only be used in theoretical derivations or in the comparison study based on the simulated data where the covariance matrix $\bSigma$ is known. In practice, $\bSigma$ is usually replaced by its estimator $\bS_{n+1,m}$ constructed by using the asset returns $\by_{n+1},...,\by_{n+m}$ from time $n+1$ to $n+m$ and defined by
\begin{equation}\label{bS_m}
\bS_{n+1:n+m}=\frac{1}{m-1} \sum_{i=n+1}^{n+m} (\by_i-\bar{\by}_{n+1:n+m})(\by_i-\bar{\by}_{n+1:n+m})^\top
\quad \text{with} \quad
\bar{\by}_{n+1:n+m}=\frac{1}{m} \sum_{i=n+1}^{n+m} \by_i.
\end{equation}
Consequently, the out-of-sample variance and the out-of-sample relative loss are replaced by the sample counterparts, the so-called empirical out-of-sample variance and the empirical out-of-sample relative loss expressed as 
\begin{equation}\label{emp_oosVar}
\hat{V}_{\hat{\bw}_n;m}= \hat{\bw}_n^\top \bS_{n+1:m+1}\hat{\bw}_n,   
\end{equation}
and
\begin{equation}\label{emp_oosLoss}
\hat{L}_{\hat{\bw}_n;m}=\frac{\hat{V}_{\hat{\bw}_n}-(1-\tilde{c})^{-1}\hat{V}_{n+1:n+m;GMV}}{(1-\tilde{c})^{-1}\hat{V}_{n+1:n+m;GMV}}=(1-\tilde{c}) \bOne^\top\bS_{n+1:m+1}^{-1}\bOne\hat{\bw}_n^\top \bS_{n+1:m+1} \hat{\bw}_n-1,  
\end{equation}
respectively, with $p/m \to \tilde{c}$ as $m \to \infty$. In \eqref{emp_oosLoss},  $(1-\tilde{c})\hat{V}_{n+1:n+m;GMV}$ is a consistent estimator for $V_{GMV}$ in the high-dimensional setting (see, Lemma 1.3 in \cite{bodnarokhrinparolya2020}).

\begin{comment}
Figure 1 depicts the results of the small simulation study, where asset returns are drawn from a multivariate normal distribution and the traditional GMV portfolio \eqref{w_S} is used in the computation of the out-of sample relative loss and its empirical counterpart. The results are present for several values of $c \in \{0.5,0.9\}$ and $\tilde{c}\in \{0.5,0.9\}$. We set $p=300$. {\color{red} Specification of $\bmu$ and $\bSigma$ to be inserted}. {\color{green} In the figure we observe considerable differences between $L_{\hat{\bw}_S}$ and $\hat{L}_{\hat{\bw}_S}$. The differences in the two performance measures are larger for $c=0.9$ independently of $\tilde{c}$. Moreover, the empirical out-of-sample relative loss becomes very smaller for $c=0.9$ in comparison to $c=0.5$, while the opposite behaviour is present in the case of the out-of-sample relative loss.}

%\begin{figure}
%#    \centering
%    \includegraphics[width=\textwidth]{figures/example_scenario_nsims50.png}
%    \caption{ Rows are in sample and columns are out of sample.}
%    \label{fig:my_label}
%\end{figure}
\end{comment}

We contribute in this paper by deriving the asymptotic behaviour of the out-of-sample variance, of the out-of-sample relative loss, and of their empirical counterparts in the high-dimensional setting, i.e., when $p/n \to c$ as $n \to \infty$ and $p/m \to \tilde{c}$ as $m \to \infty$. The results are obtained for the sample estimator \eqref{w_S} of the GMV portfolio \eqref{w_gmv}, for two shrinkage estimators introduced by \cite{frahm2010} and \cite{bodnar2018estimation}, and for the equally-weighted portfolio, which is used as a target portfolio in the specification of the considered two shrinkage estimators. We show that the empirical out-of-sample variance might tend to zero independently of chosen estimator of the GMV portfolio, which make hard to distinguish between the estimators in practice. In contrast, the empirical out-of-sample losses of the considered estimators of the GMV portfolio tend to deterministic finite quantities. As such, a decision about the ranking of the estimators can be drawn. Moreover, one needs milder conditions for the derivation of the asymptotic properties of the empirical out-of-sample relative loss in comparison to the empirical out-of-sample variance, which is an additional advantage for the application of the former in practice.

Statistical methods used in the derivation of improved estimators of optimal portfolio weights and of the performance measures are closely related to the approaches applied in statistical signal processing. In particular, the GMV portfolio is linked to the Capon or minimum variance spatial filter in signal processing literature (see, e.g., \citet{verdu1998}, \citet{vantrees2002}). \citet{rubio2012performance}, \citet{yang2018high}, \citet{LiStoicaWang2004} studied the estimation risk in the case of the high-dimensional minimum variance beamformer, while \citet{MestreLaugunas2006} investigate the finite-sample size effect on minimum variance filter. \citet{Palomar2013} discuss the improved estimation of the inverse covariance matrix from signal processing perspectives. Finally, applications of random matrix theory to signal processing and portfolio optimization are provided in \citet{PalomarBook2016}, among others.

The rest of the paper is structured as follows. In Section \ref{sec:oosv}, the asymptotic behaviour of the out-of-sample variance and of the out-of-sample relative loss is established for the traditional sample estimator and for the two shrinkage approaches. Section \ref{sec:emp_oosv} presents the corresponding results in the case of the empirical performance measures. The results of a comprehensive simulation study are provided in Section \ref{sec:sim}, while the theoretical findings are implemented to real data in Section \ref{sec:emp}. Concluding remarks are drawn in Section \ref{sec:sum}. The proofs of the theoretical results are postponed to the appendix (Section \ref{sec:app}).

\section{Out-of-sample variance and relative loss}\label{sec:oosv}

Let the vector of asset returns, $\by_{1},...,\by_{n}, \by_{n+1},...,\by_{n+m}$ be independent and identically distributed with the following stochastic representation
\begin{equation}\label{model_yi}
\by_t=\bmu+\bSigma^{1/2}\bx_t,    
\end{equation}
where the components of $\bx_t$ are independent and identically distributed with zero mean, unit variance, and finite $4+\epsilon$ moments for some $\epsilon>0$. No specific distributional assumptions are imposed on the components of $\bx_t$. The symbol $\bSigma^{1/2}$ denotes the square root of a positive definite matrix $\bSigma$, i.e., $\bSigma=\bSigma^{1/2}(\bSigma^{1/2})^\top$. Finally, we note that only $\by_t$, $t=1,...,n+m$, are observable, while $\bmu$, $\bSigma$, and $\bx_t$, $t=1,...,n+m$, are all unknown. 

Depending on the performance measure different assumptions on the covariance matrix $\bSigma$ and on the weights $\bb$ of the target portfolio are imposed. They are summarized as follows:

\begin{description}
\item[\textbf{(A1)}] The variance of the GMV portfolio $V_{GMV}$ as given in \eqref{V_gmv} is uniformly bounded in $p$.

\item[\textbf{(A2)}] The variance of the target portfolio $V_{\bb}=\bb^\top \bSigma \bb$
 is uniformly bounded in $p$.
 
\item[\textbf{(A3)}] The relative loss of the target portfolio 
\begin{equation*}
L_{\bb}=\frac{V_{\bb}-V_{GMV}}{V_{GMV}}= \bOne^\top\bSigma^{-1}\bOne\bb^\top \bSigma \bb-1,   
\end{equation*}
is uniformly bounded in $p$.
\end{description}

The considered assumptions are very general and are fulfilled in many applications. For instance, all three assumptions are fulfilled when the eigenvalues of $\bSigma$ are uniformly bounded in $p$ and the Euclidean norm of the target vector $\bb$ is uniformly bounded in $p$. Assumptions \textbf{(A1)} and \textbf{(A2)} will be needed when the out-of-sample variance \eqref{oosVar} and its empirical counterpart \eqref{emp_oosVar} are analyzed, while Assumption \textbf{(A3)} is required only in the case of the out-of-sample relative loos \eqref{oosLoss} and of the empirical out-of-sample relative loss \eqref{emp_oosLoss}. This is not surprising, since the relative loss functions are already normalized and for that reason less restrictive assumptions are needed to study their asymptotic behaviour. Furthermore, the normalization constant does not depend on an estimator of the GMV portfolio weights and thus, the normalization has no impact on the selected estimator.

Two shrinkage estimators for the GMV portfolio weights were derived in \cite{frahm2010} and \cite{bodnar2018estimation}, and they are given by
\begin{equation}\label{w_FM}
\hbw_{n;FM}=  \hat{\alpha}_{n;FM} \hat{\mathbf{w}}_{n;S} + (1- \hat{\alpha}_{n;FM})\bb
\end{equation}
with
\begin{equation}\label{alpha_GMV-FM}
\hat{\alpha}_{n;FM} = 1- \frac{p-3}{n-p+2}\left( \bOne^\top\bS_n^{-1}\bOne \bb^\top\bS_n\bb-1\right)^{-1}.
\end{equation}
and
\begin{equation}\label{w_BPS}
\hbw_{n;BPS}=\hat{\alpha}_{n;BPS} \hat{\mathbf{w}}_{n;S} + (1- \hat{\alpha}_{n;BPS})\mathbf{b}
\end{equation}
with
\begin{equation}\label{alpha_GMV-BPS}
\hat{\alpha}_{n;BPS} = \frac{\left(1-p/n\right)\left(\left(1-p/n\right) \bOne^\top\bS_n^{-1}\bOne \bb^\top\bS_n\bb-1\right)}{p/n+\left(1-p/n\right)\left(\left(1-p/n\right) \bOne^\top\bS_n^{-1}\bOne \bb^\top\bS_n\bb-1\right)},
\end{equation}
respectively.

Next, we present the asymptotic behaviour of the out-of-sample variance (Theorem \ref{th:th1}) and of the out-of-sample relative loss (Theorem \ref{th:th2}) calculated for the sample estimator $\hbw_{n;S}$ of the GMV portfolio weights and for two shrinkage estimators $\hbw_{n;FM}$ and $\hbw_{n;BPS}$ in the high-dimensional setting. The proofs of the theorems are given in the appendix. To this end, we note that the out-of-sample variance and the out-of-sample loss of the target portfolio $\bb$ are, by definition, expressed as
   \begin{equation}\label{oosV-b}
    V_{\bb}= \bb^\top \bSigma \bb
    \end{equation}
and
   \begin{equation}\label{oosL-b}
    L_{\bb}= \frac{V_{\bb}}{V_{GMV}}-1 = \bOne^\top\bSigma^{-1}\bOne \bb^\top \bSigma \bb-1,
    \end{equation}
respectively.

\begin{theorem}\label{th:th1}
Let $\by_{t}$, $t=1,...,n$ follow model \eqref{model_yi}. Then,
\begin{enumerate}[(i)]
    \item under Assumption \textbf{(A1)}, for the out-of-sample variance of the sample GMV portfolio $\hbw_{n;S}$ it holds that 
    \begin{equation}\label{th1-oosV-wS}
    \left|V_{\hbw_{n;S}}-(1-c)^{-1} V_{GMV}\right|  \stackrel{a.s.}{\rightarrow} 0,
    \end{equation}
    \item under Assumptions \textbf{(A1)} and \textbf{(A2)}, for the out-of-sample variance of the shrinkage GMV portfolio $\hbw_{n;BPS}$ it holds that
    \begin{equation}\label{th1-oosV-wBPS}
    \left|V_{\hbw_{n;BPS}}-\left(V_{GMV}+    \alpha_{BPS}^2\frac{c}{1-c}V_{GMV}+(1-\alpha_{BPS})^2(V_{\bb}-V_{GMV})\right)\right|    \stackrel{a.s.}{\rightarrow} 0
    \end{equation}
with
\begin{equation}\label{th1-oosV-alpBPS}
    \alpha_{BPS}=\frac{(1-c)L_{\bb}}{c+(1-c)L_{\bb}},
\end{equation}
\item under Assumptions \textbf{(A1)} and \textbf{(A2)}, for the out-of-sample variance of the shrinkage GMV portfolio $\hbw_{n;FM}$ it holds that
    \begin{equation}\label{th1-oosV-wFM}
    \left|V_{\hbw_{n;FM}}-\left(V_{GMV}+    \alpha_{FM}^2\frac{c}{1-c}V_{GMV}+(1-\alpha_{FM})^2(V_{\bb}-V_{GMV})\right)\right|    \stackrel{a.s.}{\rightarrow} 0
    \end{equation}
with
\begin{equation}\label{th1-oosV-alpFM}
      \alpha_{FM}=1-\frac{c}{1-c}((1-c)^{-1}(L_{\bb}+1)-1)^{-1}
      %=1-\frac{c}{L_{\bOne/p}+c}
      =\frac{L_{\bb}}{L_{\bb}+c},
\end{equation}
\end{enumerate}
for $p/n \rightarrow c \in (0, 1)$ as $n\rightarrow \infty$.
\end{theorem}

\begin{theorem}\label{th:th2}
Let $\by_{t}$, $t=1,...,n$ follow model \eqref{model_yi}. Then,
\begin{enumerate}[(i)]
    \item for the out-of-sample relative loss of the sample GMV portfolio $\hbw_{n;S}$ it holds that 
    \begin{equation}\label{th2-oosL-wS}
    \left|L_{\hbw_{n;S}}-\frac{c}{1-c}\right|    \stackrel{a.s.}{\rightarrow} 0,
    \end{equation}
\item under Assumption \textbf{(A3)}, for the out-of-sample relative loss of the shrinkage GMV portfolio $\hbw_{n;BPS}$ it holds that
    \begin{equation}\label{th2-oosL-wBPS}
    \left|L_{\hbw_{n;BPS}}-\left(\alpha_{BPS}^2\frac{c}{1-c}+(1-\alpha_{BPS})^2L_{\bb}\right)\right|    \stackrel{a.s.}{\rightarrow} 0,
    \end{equation}
\item under Assumption \textbf{(A3)}, for the out-of-sample relative loss of the shrinkage GMV portfolio $\hbw_{n;FM}$ it holds that
    \begin{equation}\label{th2-oosL-wFM}
    \left|L_{\hbw_{n;FM}}-\left(\alpha_{FM}^2\frac{c}{1-c}+(1-\alpha_{FM})^2 L_{\bb}\right)\right|    \stackrel{a.s.}{\rightarrow} 0,
    \end{equation}
    \end{enumerate}
for $p/n \rightarrow c \in (0, 1)$ as $n\rightarrow \infty$ where $\alpha_{BPS}$ and $\alpha_{FM}$ are given in \eqref{th1-oosV-alpBPS} and \eqref{th1-oosV-alpFM}, respectively.
\end{theorem}

The findings of Theorem \ref{th:th2} shows that the relative loss of shrinkage portfolios is present as a linear combination of the relative loss of the corresponding target portfolio and of the limiting relative loss of the traditional GMV portfolio. The relative loss of the traditional GMV portfolio $\hbw_{n;S}$ tends to a constant $c/(1-c)$ that does not depend on the covariance matrix of the asset returns. Moreover, if $c$ tends to $1$, then the relative loss of the traditional GMV portfolio tends to infinity showing that the impact of the estimation error could be drastically large in the high-dimensional setting. Furthermore, using \eqref{th1-oosV-alpBPS} and \eqref{th1-oosV-alpFM} the limiting values of relative loss computed for two shrinkage estimators can be rewritten as
\begin{equation}\label{asym_L_FM}
\alpha_{FM}^2\frac{c}{1-c}+(1-\alpha_{FM})^2 L_{\bb}= \frac{L_{\bb}^2}{(c+L_{\bb})^2}  \frac{c}{1-c}+\frac{c^2}{(c+L_{\bb})^2}L_{\bb} 
\end{equation}
for the shrinkage estimator of \cite{frahm2010} and
\begin{equation}\label{asym_L_BPS}
\alpha_{BPS}^2\frac{c}{1-c}+(1-\alpha_{BPS})^2 L_{\bb} =   
\frac{(1-c)L_{\bb}^2}{(c+(1-c)L_{\bb})^2} c +\frac{c^2}{(c+(1-c)L_{\bb})^2}L_{\bb} 
\end{equation}
for the shrinkage estimator of \cite{bodnar2018estimation}. As a result, expressions \eqref{asym_L_FM} and \eqref{asym_L_BPS} show that the out-of-sample relative loss of the shrinkage estimator \eqref{w_FM} tends to infinity as $c$ approaches one, similarly to the traditional estimator $\hbw_{n;S}$, while the out-of-sample relative loss of the shrinkage estimator \eqref{w_BPS} tends to the relative loss of the target portfolio when $c$ tends to one.

The results of Theorem \ref{th:th2} lead also to some dominance statements presented Corollary \ref{cor:cor1} in terms of the out-of-sample relative loss. Due to the relationship between the out-of-sample variance and the out-of-sample loss the same statements also hold for the out-of-sample variance by using the findings of Theorem \ref{th:th1}.

\begin{corollary}\label{cor:cor1}
Let $\by_{t}$, $t=1,...,n$ follow model \eqref{model_yi}. Then, under Assumption \textbf{(A3)} it holds that
\begin{enumerate}[(i)]
\item 
\[L_{\hbw_{n;S}}-L_{\hbw_{n;FM}}\stackrel{a.s.}{\rightarrow}\frac{c^2(c+L_{\bb}+cL_{\bb})}{(1-c)(c+L_{\bb})^2}\ge 0,
\text{ for $\frac{p}{n}\rightarrow c \in (0, 1)$ as $n\rightarrow \infty$,}\]
with equality if and only if $c=0$ or $L_{\bb}=\infty$, i.e., when the sample size is considerably larger than the portfolio dimension or the target portfolio deviates too strong from the true GMV portfolio;
\item \[ L_{\hbw_{n;S}}-L_{\hbw_{n;BPS}} \stackrel{a.s.}{\rightarrow}\frac{c^2}{(1-c)(c+(1-c)L_{\bb})} \ge 0
\text{ for $\frac{p}{n}\rightarrow c \in (0, 1)$ as $n\rightarrow \infty$,}\]
with equality if and only if $c=0$ or $L_{\bb}=\infty$, i.e., when the sample size is considerably larger than the portfolio dimension or the target portfolio deviates too strong from the true GMV portfolio; 
\item 
    \[L_{\hbw_{n;FM}} - L_{\hbw_{n;BPS}}\stackrel{a.s.}{\rightarrow}\frac{c^4 L_{\bb}^2}{(1-c)(c+L_{\bb})^2(c+(1-c)L_{\bb})} \ge 0
    \text{ for $\frac{p}{n}\rightarrow c \in (0, 1)$ as $n\rightarrow \infty$,}\]
    with equality if and only if  $c=0$ or $c>0$, $L_{\bb}=0$  or $c>0$, $L_{\bb}=\infty$, i.e., when the target portfolio coincides with the true GMV portfolio or the target portfolio deviates too strong from the true GMV portfolio when the concentration ratio is positive.
\end{enumerate}
\end{corollary}

The findings of Corollary \ref{cor:cor1} show that the shrinkage estimator of \cite{bodnar2018estimation} outperforms the other two estimators, while the shrinkage estimator of \cite{frahm2010} is always better than the sample estimator $\hbw_{n;S}$. The exception is present when the sample size $n$ is considerably larger than the portfolio dimension $p$ such that the concentration ratio is equal to zero or when the target portfolio is very poorly chosen such that its relative loss is infinity. In the latter situation, the investor might consider a different target portfolio in order to get the advantage of the shrinkage approaches over the sample estimator. Interestingly, when the target portfolio coincide with the population GMV portfolio, then both shrinkage estimators perform similarly.

\section{Empirical out-of-sample variance and relative loss}\label{sec:emp_oosv}
The results of Theorems \ref{th:th1} and \ref{th:th2} cannot be used in practice, since the definitions of both the out-of-sample variance and the out-of-sample relative loss depend on the unknown population covariance matrix $\bSigma$. As a result, different portfolio strategies are compared between each other based on the empirical counterparts of the out-of-sample performance measures as presented in \eqref{emp_oosVar} and \eqref{emp_oosLoss}, respectively, where the sample of the asset returns $\by_{n+1},...,\by_{n+m}$ is used to construct an estimator of the covariance matrix denoted by $\bS_{n+1:n+m}$ as in \eqref{bS_m}.

In Theorems \ref{th:th3} and \ref{th:th4} we derive the asymptotic properties of the empirical out-of-sample variance and of the empirical out-of-sample relative loss computed for the four portfolios discussed in Section \ref{sec:oosv}. The proofs of the theorems are presented in the appendix. It is remarkable that the results of Theorems \ref{th:th3} and \ref{th:th4} are deduced under the same conditions as given in the statements of Theorems \ref{th:th1} and \ref{th:th2}, even though additional randomness is taken into account in the derivations of the results. Moreover, both the empirical out-of-sample variances and the out-of-sample relative losses converge to the same limiting values as given in Theorems \ref{th:th1} and \ref{th:th2}.

\begin{theorem}\label{th:th3}
Let $\by_{t}$, $t=1,...,n+m$ follow model \eqref{model_yi}. Then,
\begin{enumerate}[(i)]
    \item under Assumption \textbf{(A1)}, for the empirical out-of-sample variance of the sample GMV portfolio $\hbw_{n;S}$ it holds that 
    \begin{equation}\label{th3-emp-oosV-wS}
    \left|\hat{V}_{\hbw_{n;S};m}-(1-c)^{-1} V_{GMV}\right|    \stackrel{a.s.}{\rightarrow} 0,
    \end{equation}
    \item under Assumption \textbf{(A2)}, for the empirical out-of-sample variance of the target portfolio $\bb$ it holds that 
   \begin{equation}\label{th3-emp-oosV-b}
    \left|\hat{V}_{\bb;m}-V_{\bb}\right|    \stackrel{a.s.}{\rightarrow} 0,
    \end{equation}
    %with $\hat{V}_{\bb;m}=\bb^\top \bS_{n+1:n+m}\bb$,
    \item under Assumptions \textbf{(A1)} and \textbf{(A2)}, for the empirical out-of-sample variance of the shrinkage GMV portfolio $\hbw_{n;BPS}$ it holds that
    \begin{equation}\label{th3-emp-oosV-wBPS}
    \left|\hat{V}_{\hbw_{n;BPS};m}-\left(V_{GMV}+    \alpha_{BPS}^2\frac{c}{1-c}V_{GMV}+(1-\alpha_{BPS})^2(V_{\bb}-V_{GMV})\right)\right|    \stackrel{a.s.}{\rightarrow} 0,
    \end{equation} 
    with $\alpha_{BPS}$ as in \eqref{th1-oosV-alpBPS},
    \item under Assumptions \textbf{(A1)} and \textbf{(A2)}, for the empirical out-of-sample variance of the shrinkage GMV portfolio $\hbw_{n;FM}$ it holds that
    \begin{equation}\label{th3-emp-oosV-wFM}
    \left|\hat{V}_{\hbw_{n;FM};m}-\left(V_{GMV}+    \alpha_{FM}^2\frac{c}{1-c}V_{GMV}+(1-\alpha_{FM})^2(V_{\bb}-V_{GMV})\right)\right|    \stackrel{a.s.}{\rightarrow} 0,
    \end{equation}
    with $\alpha_{FM}$ as in \eqref{th1-oosV-alpFM},
\end{enumerate}
for $p/n \rightarrow c \in (0, 1)$ and $p/m \rightarrow \tilde{c} \in (0, \infty)$ as $n\rightarrow \infty$.
\end{theorem}

\begin{theorem}\label{th:th4}
Let $\by_{t}$, $t=1,...,n+m$ follow model \eqref{model_yi}. Then,
\begin{enumerate}[(i)]
    \item under Assumption \textbf{(A3)}, for the empirical out-of-sample relative loss of the sample GMV portfolio $\hbw_{n;S}$ it holds that 
    \begin{equation}\label{th4-emp-oosL-wS}
    \left|\hat{L}_{\hbw_{n;S};m}-\frac{c}{1-c}\right|    \stackrel{a.s.}{\rightarrow} 0,
    \end{equation}
    \item under Assumption \textbf{(A3)}, for the empirical out-of-sample relative loss of the target portfolio $\bb$ it holds that 
   \begin{equation}\label{th4-emp-oosL-b}
    \left|\hat{L}_{\bb;m}-L_{\bb}\right|    \stackrel{a.s.}{\rightarrow} 0,
    \end{equation}
    %with $\hat{V}_{\bb;m}=\bb^\top \bS_{n+1:n+m}\bb$,
    \item under Assumptions \textbf{(A3)}, for the empirical out-of-sample relative loss of the shrinkage GMV portfolio $\hbw_{n;BPS}$ it holds that
    {\small\begin{equation}\label{th4-emp-oosL-wBPS}
    \left|\hat{L}_{\hbw_{n;BPS};m}-\left(\alpha_{BPS}^2\frac{c}{1-c}
    +(1-\alpha_{BPS})^2 L_{\bb}\right)\right|    \stackrel{a.s.}{\rightarrow} 0,
    \end{equation} }
    with $\alpha_{BPS}$ as in \eqref{th1-oosV-alpBPS},
    \item under Assumptions \textbf{(A3)}, for the empirical out-of-sample relative loss of the shrinkage GMV portfolio $\hbw_{n;FM}$ it holds that
    {\small\begin{equation}\label{th4-emp-oosL-wFM}
    \left|\hat{L}_{\hbw_{n;FM};m}-\left(\alpha_{FM}^2\frac{c}{1-c}
    +(1-\alpha_{FM})^2 L_{\bb}\right)\right|    \stackrel{a.s.}{\rightarrow} 0,
    \end{equation}}
    with $\alpha_{FM}$ as in \eqref{th1-oosV-alpFM},
\end{enumerate}
for $p/n \rightarrow c \in (0, 1)$ and $p/m \rightarrow c \in (0, 1)$ as $n,m\rightarrow \infty$.
\end{theorem}

Since the empirical out-of-sample losses $\hat{L}_{\hbw_{n;S};m}$, $\hat{L}_{\hbw_{n;BPS};m}$, and $\hat{L}_{\hbw_{n;FM};m}$ possess the same high-dimensional asymptotic behaviour as the corresponding out-of-sample losses $L_{\hbw_{n;S}}$, $L_{\hbw_{n;BPS}}$, and $L_{\hbw_{n;FM}}$ in Theorem \ref{th:th2}, the results of Corollary \ref{cor:cor1} remain also valid. Namely, we get

\begin{corollary}\label{cor:cor2}
Let $\by_{i}$, $i=1,...,n+m$ follow model \eqref{model_yi}. Then, under Assumption \textbf{(A3)} it holds that
\begin{enumerate}[(i)]
\item 
\[\hat{L}_{\hbw_{n;S;m}}-\hat{L}_{\hbw_{n;FM;m}}\stackrel{a.s.}{\rightarrow}\frac{c^2(c+L_{\bb}+cL_{\bb})}{(1-c)(c+L_{\bb})^2}\ge 0\]
for ${p}/{n}\rightarrow c \in (0, 1)$, ${p}/{m} \rightarrow \tilde{c} \in (0, 1)$ as $n,m\rightarrow \infty$,
with equality if and only if $c=0$ or $L_{\bb}=\infty$, i.e., when the sample size is considerably larger than the portfolio dimension or the target portfolio deviates too strong from the true GMV portfolio;
\item \[\hat{L}_{\hbw_{n;S;m}}-\hat{L}_{\hbw_{n;BPS;m}} \stackrel{a.s.}{\rightarrow}\frac{c^2}{(1-c)(c+(1-c)L_{\bb})} \ge 0\]
for ${p}/{n}\rightarrow c \in (0, 1)$, ${p}/{m} \rightarrow \tilde{c} \in (0, 1)$ as $n,m\rightarrow \infty$,
with equality if and only if $c=0$ or $L_{\bb}=\infty$, i.e., when the sample size is considerably larger than the portfolio dimension or the target portfolio deviates too strong from the true GMV portfolio; 
\item 
    \[\hat{L}_{\hbw_{n;FM;m}} - \hat{L}_{\hbw_{n;BPS;m}}\stackrel{a.s.}{\rightarrow}\frac{c^4 L_{\bb}^2}{(1-c)(c+L_{\bb})^2(c+(1-c)L_{\bb})} \ge 0\]
for ${p}/{n}\rightarrow c \in (0, 1)$, ${p}/{m} \rightarrow \tilde{c} \in (0, 1)$ as $n,m\rightarrow \infty$, with equality if and only if  $c=0$ or $c>0$, $L_{\bb}=0$  or $c>0$, $L_{\bb}=\infty$, i.e., when the target portfolio coincides with the true GMV portfolio or the target portfolio deviates too strong from the true GMV portfolio when the concentration ratio is positive.
\end{enumerate}
\end{corollary}

Corollary \ref{cor:cor2} provides the limiting behaviour of the differences of the empirical out-of-sample losses and, consequently, the same ranking between the three estimators of the GMV portfolio weights as previously obtained in Corollary \ref{cor:cor1}. Furthermore, the difference between the asymptotic behaviour of the three estimator is negligible only when the concentration ratio is zero, i.e., the portfolio size is considerably smaller than the sample size, or when the target portfolio is poorly chosen such that its relative loss becomes infinity.

%%%%%%%%%%%%%%%%%%%%%%%%%%%%%%%%%%%%
% limits
The asymptotic differences between the relative losses of these three estimators are depicted as functions in $c \in (0,1)$ for several values of $L_{\bb} \in (0,50)$ in Figure \ref{fig:cor_limits1}. Larger differences are observed when the shrinkage estimator of \cite{bodnar2018estimation} is compared to the traditional estimator and the shrinkage estimator of \cite{frahm2010}, especially when $c$ is close to one. On the other side, the asymptotic difference between the empirical out-of-sample relative loss functions computed for the traditional estimator and the shrinkage estimator of \cite{frahm2010} is large only when $L_b$ is close to zero, i.e., when the target portfolio $\bb$ is close to the true population GMV portfolio.

%In Figure we can see the limits of Corollary \ref{cor:cor2}. In these figures, we have evaluated the limits for $c\in (0,1)$ and $L_b\in (0, 50)$ with an increment of $0.01$. Notice that the scales are different, which is a consequence of the natural ordering Corollary \ref{cor:cor2} provides.

\begin{figure}
    \centering
    \includegraphics[width=\textwidth]{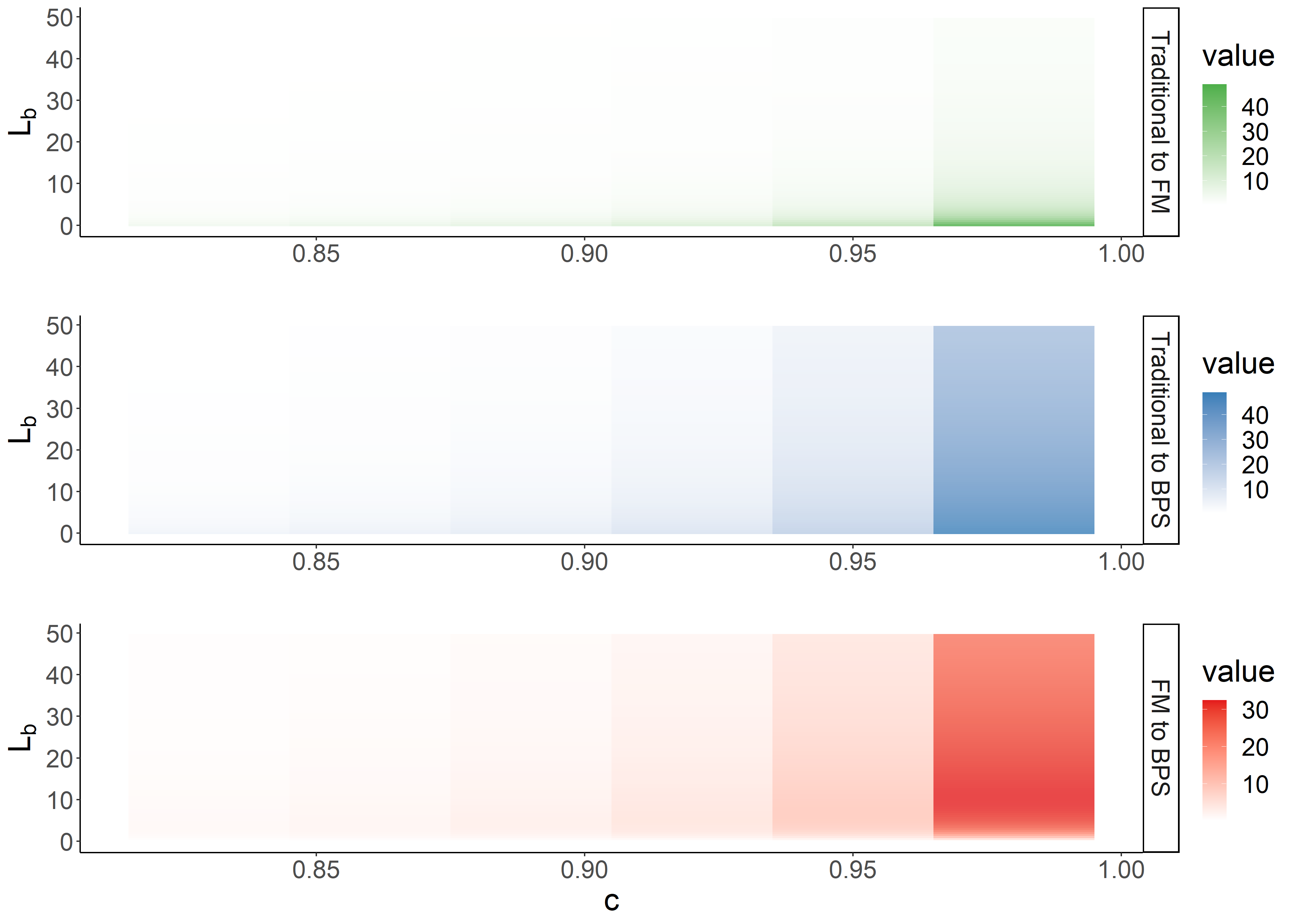}
    \caption{Asymptotic differences between the empirical out-of-sample relative losses limits from Corollary \ref{cor:cor2} for $c \in (0,1)$ and $L_b \in (0,50)$.}
    \label{fig:cor_limits1}
\end{figure}

%%%%%%%%%%%%%
% 3D limit.
%\begin{figure}
%    \centering
    %\includegraphics[width=\textwidth]{figures/limit_fig_3d_2.png}
    %\caption{Asymptotic differences between the empirical out-of-sample relative losses limits from Corollary \ref{cor:cor2} for $c \in (0,1)$ and $L_b \in (0,50)$.}
%    \label{fig:cor_limits2}
%\end{figure}

\section{Simulation study}\label{sec:sim}
In this section we will investigate the finite sample behaviour of the high-dimensional asymptotic results presented in Corollary \ref{cor:cor2} via an extensive Monte Carlo study. The aim of the study is twofold: (i) first, we investigate how fast the difference of the empirical out-of-sample relative loss functions tend to the corresponding limiting value provided in the statement of Corollary \ref{cor:cor2}; (ii) second, we study the impact of the presence of linear and non-linear time dependence in the data-generating model on the performance of the three considered trading strategies.

For each fixed value of the portfolio size $p$ we first simulated the elements of the mean vector $\bmu$ as $\mu_i \sim U(-0.1,0.1)$, $i=1,2,...,p$ and the elements of the covariance matrix $\bSigma$ using the \texttt{RandCovMtrx} function from the \texttt{HDShOP} package (\cite{HDShOP}). Then these values were used in simulating samples of the asset returns from the following three data-generating models:

\begin{description}
    \item[Scenario 1: $t$-distribution]
    The elements of $\bx_t$ are drawn independently from the $t$-distribution with $5$ degrees of freedom, that is $x_{tj}\sim t(5)$ for $j=1,...,p$, while $\by_t$ is constructed according to  \eqref{model_yi}. Moreover since the variance of the t-distribution with $5$ degrees of freedom is equal to $5/3$ we, additionally multiply the vector $\bx_t$ in \eqref{model_yi} by $\sqrt{3/5}$. As such, all $\sqrt{3/5}x_{tj}$ have mean zero and variance one.
    \item[Scenario 2: VAR model] 
    The vector of asset returns $\by_t$ is simulated according to a 
    $$
    \by_{t} = \bmu + \mathbf{\Gamma} (\by_{t-1}-\bmu) + \bSigma^{1/2}\bx_t
    \quad \text{with} \quad 
    \bx_t\sim N_p(\mathbf{0},\mathbf{I})
    $$
    for $t=1,...,n+m$, where $\mathbf{\Gamma} = \operatorname{diag}(\gamma_1, \gamma_2,..., \gamma_p)$ with $\gamma_i \sim U(-0.9, 0.9)$ for $i=1,...,p$. We note that in the case of the VAR model, the covariance matrix of $\by_t$ is computed as $\text{vec}(\mathbb{V}ar(\by))=(\bI-\mathbf{\Gamma}\otimes \mathbf{\Gamma})^{-1}\text{vec}(\bSigma)$ where $\text{vec}$ denotes the vec operator. This matrix is used in the computation of the limiting differences from Corollary \ref{cor:cor2}.
    \item[Scenario 3: CCC-GARCH model of \cite{bollerslev1990modelling}] The asset returns are simulated according to 
    $$\by_t | \bSigma_t \sim N_p(\bmu, \bSigma_t)$$
    where the conditional covariance matrix is specified by 
    $$\bSigma_t = \bD_t^{1/2} \bC \bD_t^{1/2}
    \quad \text{with} \quad \bD_t = \operatorname{diag}(h_{1,t}, h_{2,t}, ..., h_{p,t}),$$
    with
    $$
    h_{j,t} = \alpha_{j,0} + \alpha_{j,1} (\by_{j, t-1} - \bmu_j)^2 + \beta_{j,1} h_{j, t-1}, \text{ for } j=1,2,...,p, \text{ and } t=1,2,...,n+m.
    $$
    The coefficients of the CCC-GARCH model are generated by $\alpha_{j,1} \sim U(0,0.1)$ and $\beta_{j,1} \sim U(0.6,0.7)$ which implies that the stationarity conditions, $\alpha_{j,1} + \beta_{j,1} < 1$, are always fulfilled. The intercepts $\alpha_{j,0}$, $j=1,...,p$ is thereafter chosen such that the unconditional covariance matrix is equal to $\bSigma$.
\end{description}

%All three data-generating models have the same mean vector $\bmu$ and covariance matrix $\bSigma$. 
The model under scenario 1 fulfills the assumptions imposed in Section 2 by drawing the vector $\bx_t$ independently each of other. In contrast, scenarios 2 and 3 possess some time dependence structure, thus violating the assumption imposed on the data-generating model in Section \ref{sec:oosv}. While the VAR model from scenario 2 is used to investigate the performance of three portfolio selection strategies when the asset returns ${\by_t}$ are assumed to be autocorrelated, a  more complicated non-linear time dependence structure is assumed in scenario 3 which is accompanied with conditionally time-dependent covariance matrix $\bSigma_t$. Finally, the equally weighted portfolio is used as a target portfolio in all scenarios.

In Figures \ref{fig:scenario1_cor1} to \ref{fig:scenario3_cor1} we present the relative differences of empirical out-of-sample losses as considered in Corollary \ref{cor:cor2} divided by the corresponding asymptotic limit determined for each difference in the statement of the corollary in the right hand-side of each inequality. For each scenario we set $n=\{100, 250, 500, 750, 1000\}$, $c=\{0.5, 0.9\}$ and $\tilde{c}=\{0.5, 0.9\}$. The portfolio size $p$ and the sample size $m$ are thereafter determined by $p=nc$ and in turn $m=p/\tilde{c}$. If necessary we round to the closest integer. The results in the figures are based on the 1000 independent repetitions and present the corresponding average values.

%%%%%%%%%%%%%%%%%%%%%%%%%%%%%%%%%%%%
% Scenario 1
Figure \ref{fig:scenario1_cor1} depicts the results of the simulation study obtained under scenario 1. The relative differences in the empirical out-of-sample losses converge quickly to one, indicating that the results of Corollary \ref{cor:cor2} may also be used when samples of asset returns of moderate size are used. As expected, the fastest convergence is observed in the case $c=\tilde{c}=0.5$, while the largest deviations from one is present in the case of $c=0.5$ and $\tilde{c}=0.9$, when the sample size is small. Finally, we note that all computed values in the plots are positive and, as such, the shrinkage estimator of \cite{bodnar2018estimation} outperforms the other two trading strategies followed by the shrinkage approach of \cite{frahm2010} in all of the considered cases.

\begin{figure}
    \centering
    \includegraphics[width=\textwidth]{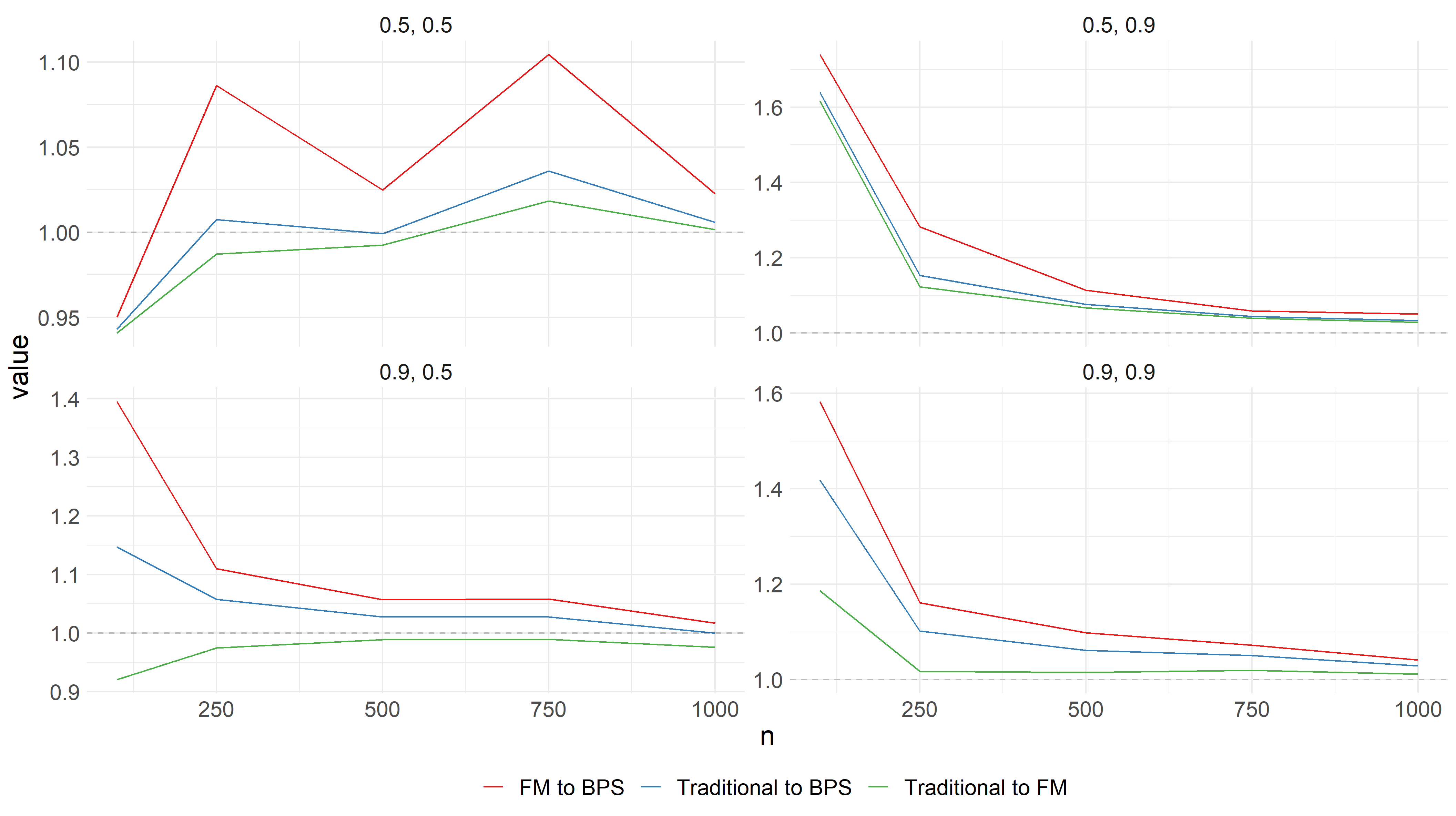}
    \caption{Relative differences in the empirical out-of-sample losses divided by the corresponding asymptotic lower bound as given in Corollary \ref{cor:cor2} for $n=\{100, 250, 500, 750, 1000\}$, $c=\{0.5, 0.9\}$ and $\tilde{c}=\{0.5, 0.9\}$. The samples of asset returns are drawn following scenario 1.}
    \label{fig:scenario1_cor1}
\end{figure}

%%%%%%%%%%%%%%%%%%%%%%%%%%%%%%%%%%%%
% Scenario 2
In Figure \ref{fig:scenario2_cor1} the results of the simulation study obtained under scenario 2 are present. This scenario imposes linear time dependence structure on the vector of asset returns and, thus, it breaks the model assumption that Corollary \ref{cor:cor2} is derived from. This can also be seen in the computed relative differences of losses. In contrast to the values shown in Figure \ref{fig:scenario1_cor1} the empirical out-of-sample relative losses do not converge to one in Figure \ref{fig:scenario2_cor1}. This indicates that the presence of linear time dependencies has an impact on the limiting properties on the empirical out-of-sample loss functions. On the other hand, the relative differences depicted in Figure \ref{fig:scenario2_cor1} are all positive and thus the ranking between the three estimation strategies remains unchanged. Moreover, the relative differences converge to the values which are larger than one, meaning that the derived limiting values in Corollary \ref{cor:cor2} can still be employed as lower bounds.

\begin{figure}
    \centering
    \includegraphics[width=\textwidth]{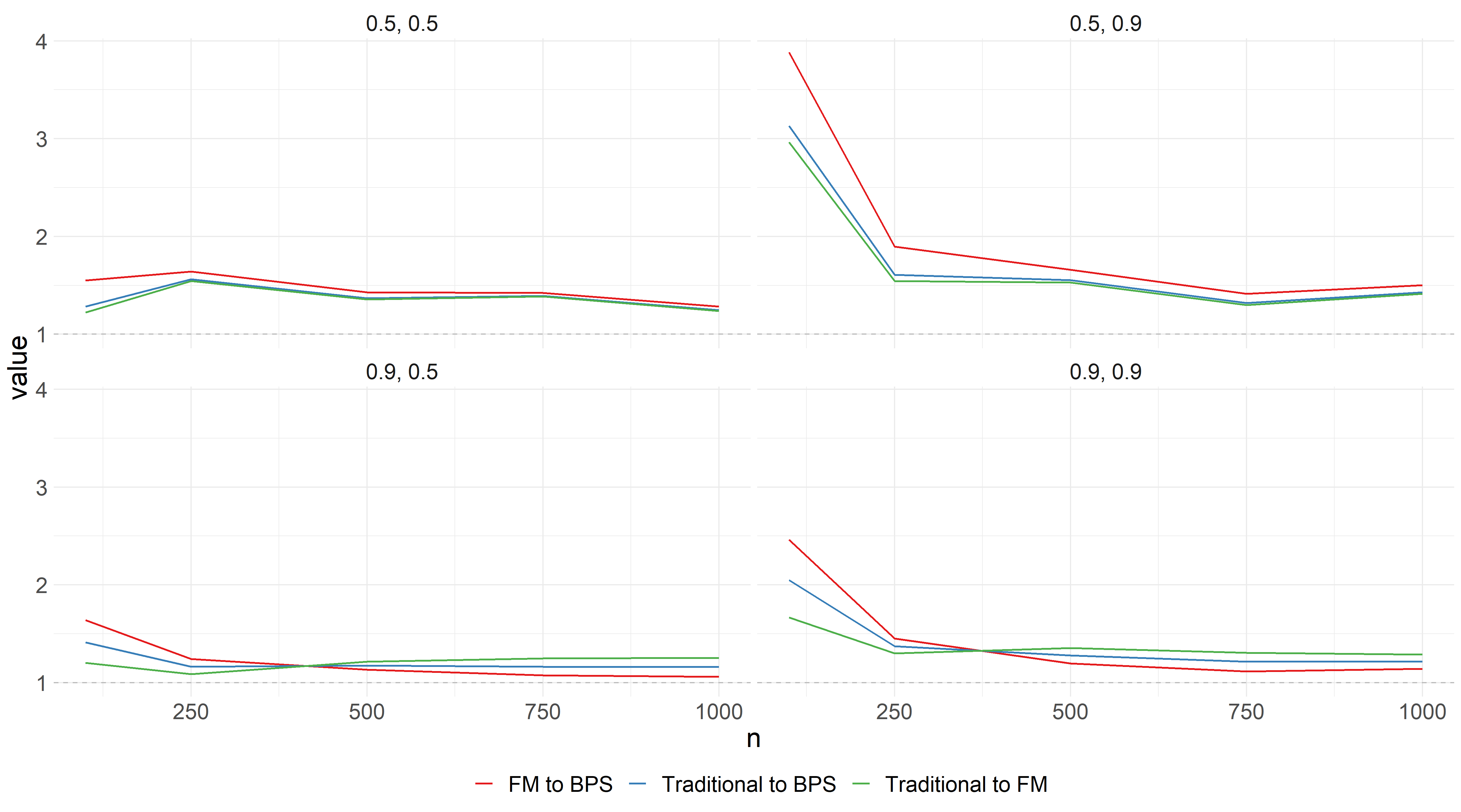}
    \caption{Relative differences in the empirical out-of-sample losses divided by the corresponding asymptotic lower bound as given in Corollary \ref{cor:cor2} for $n=\{100, 250, 500, 750, 1000\}$, $c=\{0.5, 0.9\}$ and $\tilde{c}=\{0.5, 0.9\}$. The samples of asset returns are drawn following scenario 2.}
    \label{fig:scenario2_cor1}
\end{figure}

%%%%%%%%%%%%%%%%%%%%%%%%%%%%%%%%%%%%
% Scenario 3
Figure \ref{fig:scenario3_cor1} illustrates the results of the simulation study under the last scenario. In this setting the returns are simulated from a CCC-GARCH model which captures volatility clustering and also introduces a non-linear time dependence structure in the vectors of the asset returns. Similarly to scenario 2, the relative differences do not converge to one, although the departure from one is considerably smaller as observed in the case of scenario 2. As such, a conclusion can be drawn that the presence of linear time dependence structure has larger impact on the asymptotic behaviour of the empirical out-of-sample losses than the non-linear one. Also, in scenario 3, the relative losses converge to the values which are larger one and the computed values are all positive. As such, the ranking between the three trading strategies is preserved and one can also us the expression of the limiting values of Corollary \ref{cor:cor2} as the corresponding lower bounds for the differences under the assumption of the CCC-GARCH model.

\begin{figure}
    \centering
    \includegraphics[width=\textwidth]{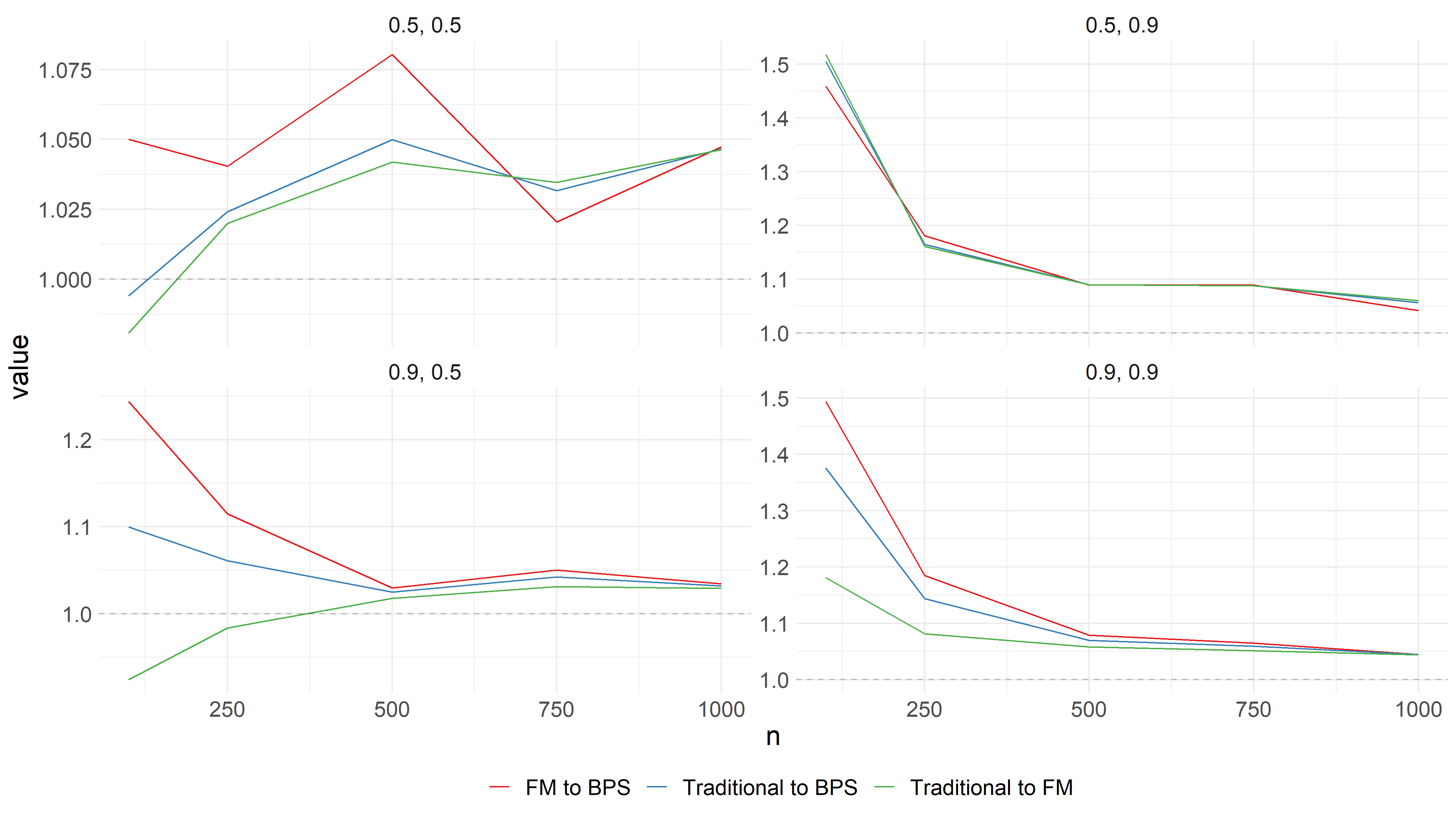}
    \caption{Relative differences in the empirical out-of-sample losses divided by the corresponding asymptotic lower bound as given in Corollary \ref{cor:cor2} for $n=\{100, 250, 500, 750, 1000\}$, $c=\{0.5, 0.9\}$ and $\tilde{c}=\{0.5, 0.9\}$. The samples of asset returns are drawn following scenario 3.}
    \label{fig:scenario3_cor1}
\end{figure}

\section{Empirical illustration}\label{sec:emp}
In the empirical application we use 10 years of daily data for 100 and 190 stocks included in the S\&P500 index from the first of June 2011 to the seventh of January 2021. During the considered period of time, 380 stocks were continuously included in the the S\&P500 index from which we randomly choose 100 and 190 stocks to build the GMV portfolio. The first $n=200$ observations were used to estimate the weights of the GMV portfolio by employing the traditional estimator and the two shrinkage estimators introduced in Section \ref{sec:oosv}, while the next $m=200$ observations were used to compute the values of the empirical out-of-sample variances and the empirical out-of-sample relative loses for each trading strategy. Then, using the rolling window approach the same computations are subsequently performed over the time period from the fourteenth of February, 2013 to the seventh of January 2021. As a target portfolio in the construction of the two shrinkage estimator, the equally weighted portfolio was used.

Figure \ref{fig:overlap} depicts the values of the empirical out-of-sample variances and of the empirical out-of-sample relative losses computed for three estimators of the GMV portfolio considered in the paper. The result are presented for two portfolio sizes which correspond to $c=\tilde{c}=0.5$ and $c=\tilde{c}=0.95$. A considerable increase in both the empirical out-of sample variances and losses of each estimator is observed in March 2020 which corresponds to the crisis on international financial market caused by the beginning of COVID-19 spread over the world. The rapid increase of volatility is more pronounced in the case of the smaller dimensional portfolio, i.e., when $p=100$. In the case of the portfolio which is based on $p=190$ stocks the jump in the values of the two considered performance measures is smoothed due to higher variability of these two measures presented during the whole period of observation. Another rapid increase in the loss functions for $p=100$ occurs in late December 2020. This date can be related to the second wave of the COVID-19 spread. Similar increases in the behaviour of the relative loss function are also present for the portfolio consisting of $p=190$ stocks, although they are somehow hidden by the more volatile behavior of the loss function in the latter case.

In general, the results in Figure \ref{fig:overlap} confirms the ordering of the three trading strategies which is deduced in Corollary \ref{cor:cor2} and confirmed in the finite-sample case in the simulation study of Section \ref{sec:sim}. Namely, the shrinkage estimator of \cite{bodnar2018estimation} shows the smallest values of both the empirical out-of-sample variance and the empirical out-of-sample relative loss, while the shrinkage estimator of \cite{frahm2010} is ranked on the second place. On the other side, when the empirical out-of-sample variance is used as a performance measure, the distinction between the strategies become visually negligible in almost all cases presented for $p=100$ and in majority of cases when the portfolio with $p=190$ is constructed. This empirical finding can be explained by noting that most of the values of the empirical out-of-sample variance were computed during the stable period on the capital market and as such, the true value of the global minimum variance was very small at that time. In contrast, the usage of the empirical out-of-sample loss can lead to the obvious conclusion about the performance of each of the considered three trading strategies. Finally, the impact of portfolio dimensionality which is accompanied with a huge amount of estimation error becomes more pronounced when the empirical relative loss is used, especially during the turbulent period on the capital market.  

\begin{figure}
    \centering
    \includegraphics[width=16cm]{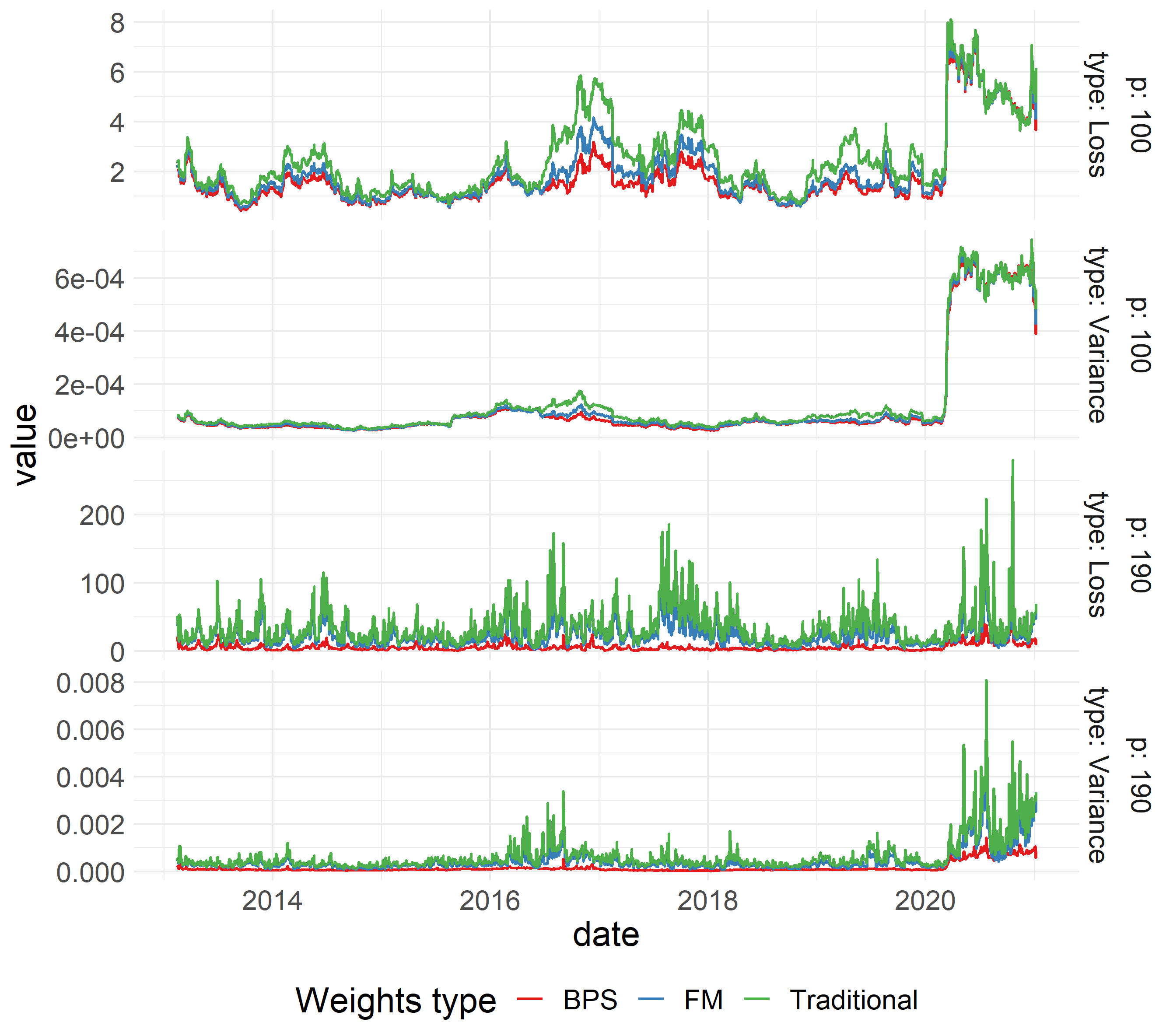}
    \caption{Empirical out-of-sample variance and out-of-sample relative loss of the the traditional GMV portfolio and the two shrinkage estimators based on the rolling window approach with window size equal $200$ and computed for two portfolios which consist of $100$ and $190$ stocks traded in the S\&P 500 index.}
    \label{fig:overlap}
\end{figure}

\begin{comment}
The first observation to make is that the ordering described by theorem \ref{th:th3} is confirmed. For the two portfolio sizes and losses, BPS is smaller than FM which is smaller than the transitional estimator. As $p$ grows we can see that the out-of-sample loss becomes more stable. The traditional variance estimator fluctuates a lot more in comparison to the smaller portfolio. The diversification effect does not dominate the variance. Both variances and losses increase with portfolio size. 
%As $p$ increases there also seem to be some stability to the losses which does not hold for the variance. 
%\textbf{Do we really want to have one in log scale and the other not in log-scale? See figures\/OOS\_moving\_window200\_logscale.png}
\end{comment}

\section{Summary}\label{sec:sum}
The sample variance of the GMV portfolio is known to be biased and to significantly underestimate the true population variance of this portfolio, especially when the portfolio size is comparable to the sample size. In many practical situations it is not a good measure for the portfolio performance and the out-of-sample variance is usually used instead. 

In this paper we derive the asymptotic properties of the out-of-sample variance and of the out-of-sample relative loss as well as of their empirical counterparts. Under weak conditions imposed on the data-generating model it is shown that the out-of-sample variance and the empirical out-of-sample variance might tend to zero independently of chosen estimator of the GMV portfolio weights, which can make the comparison between the trading strategies intractable. This is not, however, an issue when the out-of-sample relative loss and the empirical out-of-sample relative loss are used instead. In the latter case a clear ordering between the estimators of the three considered estimator can be made.

As a by product of the derived theoretical findings, we also prove that the shrinkage estimator of \cite{bodnar2018estimation} outperforms the shrinkage estimator of \cite{frahm2010} and the traditional estimator of the GMV portfolio. Moreover, we quantify the difference in the performance of the three trading strategies by deducing the asymptotic difference of their empirical out-of-sample relative loss functions. Within a comprehensive numerical study it is shown that the derive asymptotic limits can still be used when the sample of moderate size is present and when the asset returns possess both linear and non-linear time dependence structure.

\section{Appendix}\label{sec:app}

\begin{proof}[Proof of Theorem \ref{th:th1}]
\begin{enumerate}[(i)]
    \item It holds that
\[V_{\hbw_{n;S}}=\hbw_{n;S}^\top \bSigma \hbw_{n;S}=\frac{\bOne^\top\bS_n^{-1} \bSigma\bS_n^{-1}\bOne}{(\bOne^\top\bS_n^{-1}\bOne)^2},\]
where (see, proof of Lemma 1.3 in \cite{bodnarokhrinparolya2020})
\begin{eqnarray*}
|\bOne^\top\bS_n^{-1} \bSigma\bS_n^{-1}\bOne-(1-c)^{-3}\bOne^\top \bSigma^{-1}\bOne| \stackrel{a.s.}{\rightarrow} 0,\\
|\bOne^\top\bS_n^{-1}\bOne-(1-c)^{-1}\bOne^\top \bSigma^{-1}\bOne| \stackrel{a.s.}{\rightarrow} 0,
\end{eqnarray*}
for $p/n \to c \in (0,1)$ as $n \to \infty$. Combining these two results we get the first statement of the theorem.
    \item It holds that
\begin{eqnarray*}
V_{\hbw_{n;BPS}}&=&\left(\hat{\alpha}_{n;BPS} \hat{\mathbf{w}}_{n;S} + (1- \hat{\alpha}_{n;BPS})\bb\right)^\top \bSigma \left(\hat{\alpha}_{n;BPS} \hat{\mathbf{w}}_{n;S} + (1- \hat{\alpha}_{n;BPS})\bb\right)\\
&=&\hat{\alpha}_{n;BPS}^2\hat{\mathbf{w}}_{n;S}^\top\bSigma \hat{\mathbf{w}}_{n;S}+2\hat{\alpha}_{n;BPS} (1-\hat{\alpha}_{n;BPS} )\hat{\mathbf{w}}_{n;S}^\top\bSigma\bb+(1-\hat{\alpha}_{n;BPS} )^2\bb^\top\bSigma\bb,
\end{eqnarray*}
where from part (i)  
\[|\hat{\mathbf{w}}_{n;S}^\top\bSigma \hat{\mathbf{w}}_{n;S}-(1-c)^{-1}V_{GMV}| \stackrel{a.s.}{\rightarrow} 0
\quad \textbf{for $p/n \to c \in (0,1)$ as $n \to \infty$.} \]
Moreover, we get (see, Theorem 2.1 in \cite{bodnar2018estimation}) 
\[\hat{\alpha}_{n;BPS}\stackrel{a.s.}{\rightarrow}\alpha_{BPS}=\frac{(1-c)L_{\bb}}{c+(1-c)L_{\bb}}\]  
and
\[\hat{\mathbf{w}}_{n;S}^\top\bSigma\bb
=\frac{\bOne^\top \bS_n^{-1}\bSigma\bb}{\bOne^\top\bS_n^{-1}\bOne}
\stackrel{a.s.}{\rightarrow} \frac{(1-c)^{-1}}{(1-c)^{-1}\bOne^\top \bSigma^{-1}\bOne}=V_{GMV}
\]
for $p/n \to c \in (0,1)$ as $n \to \infty$. Putting these results together we get the statement of Theorem \ref{th:th1}.(ii).
\item The result of part (iii) follows from the proofs of parts (i) and (ii).
\end{enumerate}
\end{proof}

\begin{proof}[Proof of Theorem \ref{th:th2}]
The results of Theorem \ref{th:th2} follows from Theorem \ref{th:th1} and the definition of the relative loss.
\end{proof}

\vspace{0.5cm}
In the proofs of Theorems \ref{th:th3} and \ref{th:th4} we use the results of two technical lemmas presented below. Let
\begin{eqnarray*}\widetilde{\bV}_{1:n}&=&\frac{1}{n}\bX_{1:n}\bX_{1:n}^\top\quad \text{with} \quad
\bX_{1:n}= (\bx_{1},...,\bx_{n}),\\
\widetilde{\bV}_{n+1:n+m}&=&\frac{1}{m}\bX_{n+1:n+m}\bX_{n+1:n+m}^\top\quad \text{with} \quad
\bX_{n+1:n+m}= (\bx_{n+1},...,\bx_{n+m}),
\end{eqnarray*}
and define
\[\bV_{1:n}=\frac{1}{n-1}\bX_{1:n}\bX_{1:n}^\top - \frac{n}{n-1} \bar{\bx}_{1:n}\bar{\bx}_{1:n}^\top \quad \text{with} \quad
\bar{\bx}_{1:n}=\frac{1}{n}\bX_{1:n} \bOne_n
\]
and 
\[\bV_{n+1:n+m}=\frac{1}{m-1}\bX_{n+1:n+m}\bX_{n+1:n+m}^\top - \frac{m}{m-1} \bar{\bx}_{n+1:n+m}\bar{\bx}_{n+1:n+m}^\top, \;
\bar{\bx}_{n+1:n+m}=\frac{1}{m}\bX_{n+1:n+m} \bOne_m.
\]

Then, we have
\begin{equation}\label{SV-app}
\bS_{1:n}=\bSigma^{1/2}\bV_{1:n}\bSigma^{1/2}\; \text{and} \;
\bS_{n+1:n+m}=\bSigma^{1/2}\bV_{n+1:n+m}\bSigma^{1/2}.
\end{equation}

\begin{lemma}\label{lem:lem1}
Let $\bxi$ and $\btheta$ be two nonrandom vectors with bounded Euclidean norms. 
Assume that $m,n>1$. Then it holds that
    \begin{align}
        \left|\bxi^\top \widetilde{\bV}_{n}^{-1} \widetilde{\bV}_{n+1:n+m}
        \widetilde{\bV}_{n}^{-1}\btheta -(1-c)^{-3}\bxi^\top\btheta  \right| &\stackrel{a.s.}{\rightarrow} 0, \label{eqn:lem1eq1}
    \end{align}
and
    \begin{align}
        \left|\bxi^\top \widetilde{\bV}_{n}^{-1} \widetilde{\bV}_{n+1:n+m}
        \btheta -(1-c)^{-1}\bxi^\top\btheta  \right| &\stackrel{a.s.}{\rightarrow} 0, \label{eqn:lem1eq2}
    \end{align}
for $p/n \rightarrow c \in (0, 1)$ and $p/m \rightarrow \tilde{c} \in (0, \infty)$
    as $n\rightarrow \infty$.
\end{lemma}

\begin{proof}[Proof of Lemma \ref{lem:lem1}]
It holds that
\begin{eqnarray*}
&&\left|\bxi^\top \widetilde{\bV}_{n}^{-1} \widetilde{\bV}_{n+1:n+m}\widetilde{\bV}_{n}^{-1}
        \btheta - (1-c)^{-3}\bxi^\top\btheta  \right|\\
&\le&
\left|\bxi^\top \widetilde{\bV}_{n}^{-1} \widetilde{\bV}_{n+1:n+m} \widetilde{\bV}_{n}^{-1}
        \btheta - \bxi^\top \widetilde{\bV}_{n}^{-2} \btheta  \right|+
\left|\bxi^\top \widetilde{\bV}_{n}^{-2} \btheta -(1-c)^{-3}\bxi^\top\btheta  \right|, 
\end{eqnarray*}
where 
\[\left|\bxi^\top \widetilde{\bV}_{n}^{-2} \btheta -(1-c)^{-3}\bxi^\top\btheta  \right|\stackrel{a.s.}{\rightarrow} 0\]
for $p/n \rightarrow c \in (0, 1)$ as $n\rightarrow \infty$ by applying Lemma 1.3 in \cite{bodnarokhrinparolya2020}.

Furthermore, using the equality
\[
\bxi^\top \widetilde{\bV}_{n}^{-1} \widetilde{\bV}_{n+1:n+m} \widetilde{\bV}_{n}^{-1}
        \btheta=\frac{1}{m} \sum_{j=n+1}^{n+m} \bx_j^\top \widetilde{\bV}_{n}^{-1}\btheta \bxi^\top \widetilde{\bV}_{n}^{-1} \bx_j
\]
and the fact that $\widetilde{\bV}_{n}^{-1}\btheta \bxi^\top \widetilde{\bV}_{n}^{-1}$ possesses the bounded trace norm which is asymptotically bounded by $\sqrt{\btheta^\top \widetilde{\bV}_{n}^{-2}\btheta} \sqrt{ \bxi^\top \widetilde{\bV}_{n}^{-2}\bxi}$,
the application of Lemma 4 in \cite{rubio2011spectral} leads to
\[\left|\bxi^\top \widetilde{\bV}_{n}^{-1} \widetilde{\bV}_{n+1:n+m}\widetilde{\bV}_{n}^{-1}
\btheta - \bxi^\top \widetilde{\bV}_{n}^{-2}\btheta\right|\stackrel{a.s.}{\rightarrow} 0,\]
for $p/m \rightarrow \tilde{c} \in (0, \infty)$ as $m\rightarrow \infty$ for any large enough $n$. The second statement \eqref{eqn:lem1eq2} can similarly be proved. This completes the proof of the lemma.
\end{proof}

\begin{lemma}\label{lem:lem3}
Let $\bxi$ and $\btheta$ be two nonrandom vectors with bounded Euclidean norms. 
Assume that $m,n>1$. Then it holds that
    \begin{align}
        \left|\bxi^\top \bV_{n}^{-1} \bV_{n+1:n+m} \bV_{n}^{-1}\btheta - (1-c)^{-3}\bxi^\top\btheta  \right| &\stackrel{a.s.}{\rightarrow} 0, \label{eqn:lem3eq1}
    \end{align}
and
\begin{align}
        \left|\bxi^\top \bV_{n}^{-1} \bV_{n+1:n+m} \btheta - (1-c)^{-1}\bxi^\top\btheta  \right| &\stackrel{a.s.}{\rightarrow} 0, \label{eqn:lem3eq2}
    \end{align}
for $p/n \rightarrow c \in (0, 1)$ and $p/m \rightarrow \tilde{c} \in (0, \infty)$
    as $n,m\rightarrow \infty$.
\end{lemma}

\begin{proof}[Proof of Lemma \ref{lem:lem3}]
The application of the Sherman–Morrison formula leads to
\begin{eqnarray*}
&&\bxi^\top \bV_{n}^{-1} \bV_{n+1:n+m} \bV_{n}^{-1}\btheta = \frac{(n-1)^2m}{n^2(m-1)}
\bxi^\top \left(\widetilde{\bV}_{n}-\bar{\bx}_n\bar{\bx}_n^\top\right)^{-1}
\widetilde{\bV}_{n+1:n+m}
\left(\widetilde{\bV}_{n}-\bar{\bx}_n\bar{\bx}_n^\top\right)^{-1}\btheta\\
&-&  \frac{(n-1)^2m}{n^2(m-1)}
\bxi^\top \bV_{n}^{-1}\bar{\bx}_{n+1:n+m}\bar{\bx}_{n+1:n+m}^\top \bV_{n}^{-1} \btheta
\\
&=& \frac{(n-1)^2m}{n^2(m-1)}\Bigg(\bxi^\top \widetilde{\bV}_{n}^{-1} \widetilde{\bV}_{n+1:n+m}\widetilde{\bV}_{n}^{-1}\btheta
+2\frac{
\bxi^\top \widetilde{\bV}_{n}^{-1}\bar{\bx}_{n}\bar{\bx}_{n}^\top \widetilde{\bV}_{n}^{-1}\widetilde{\bV}_{n+1:n+m}\widetilde{\bV}_{n}^{-1}\btheta}{1-\bar{\bx}_{n}^\top\widetilde{\bV}_{n}^{-1}\bar{\bx}_{n}}\\
&+&
\frac{
\bxi^\top \widetilde{\bV}_{n}^{-1}\bar{\bx}_{n}\bar{\bx}_{n}^\top \widetilde{\bV}_{n}^{-1} \widetilde{\bV}_{n+1:n+m}\widetilde{\bV}_{n}^{-1}\bar{\bx}_{n}\bar{\bx}_{n}^\top\widetilde{\bV}_{n}^{-1}\btheta}{(1-\bar{\bx}_{n}^\top\widetilde{\bV}_{n}^{-1}\bar{\bx}_{n})^2}-
\bxi^\top \bV_{n}^{-1}\bar{\bx}_{n+1:n+m}\bar{\bx}_{n+1:n+m}^\top \bV_{n}^{-1} \btheta\Bigg)
\end{eqnarray*}
By definition $\sqrt{m}\bar{\bx}_{n+1:n+m}$ consists of elements with are independent and identically distributed with zero mean and variance equal one. Then, conditionally on $\bX_{1:n}$ it holds that (see, Theorem in \cite{dette2020likelihood})
\[\sqrt{m}\bar{\bx}_{n+1:n+m}\bV_{n}^{-1} \btheta|\bX_{1:n}\stackrel{d.}{\rightarrow}
\mathcal{N}\left(0,\btheta^\top \bV_{n}^{-2} \btheta\right)
\quad \textbf{as} \quad p \rightarrow \infty
\]
and, consequently, $\bar{\bx}_{n+1:n+m}^\top \bV_{n}^{-1} \btheta \stackrel{a.s.}{\rightarrow}0$ for $p/m \rightarrow\tilde{c}$ as $m \rightarrow \infty$. Finally, the applications of Lemma 5.2 in \citet{bodnarokhrinparolya2020} and Lemma \ref{lem:lem1} completes the proof of the lemma. Similarly, the result \eqref{eqn:lem3eq2} is deduced. 
\end{proof}

\begin{proof}[Proof of Theorem \ref{th:th3}]
\begin{enumerate}[(i)]
    \item We get with \eqref{SV-app} that
    \begin{eqnarray*}
    \hat{V}_{\hbw_{n;S};m}&=&\hbw_{n;S}^\top\bS_{n+1:n+m} \hbw_{n;S}
    =\frac{\bOne^\top\bS_n^{-1} \bS_{n+1:n+m}\bS_n^{-1}\bOne}{(\bOne^\top\bS_n^{-1}\bOne)^2}\\
    &=&\frac{\bOne^\top\bSigma^{-1/2}\bV_n^{-1} \bV_{n+1:n+m}\bV_n^{-1}\bSigma^{-1/2}\bOne}{(\bOne^\top\bSigma^{-1/2}\bV_n^{-1}\bSigma^{-1/2}\bOne)^2}
    \stackrel{a.s.}{\rightarrow} \frac{(1-c)^{-3}\bOne^\top\bSigma^{-1}\bOne}{(1-c)^{-2}(\bOne^\top\bSigma^{-1}\bOne)^2}=(1-c)^{-1}V_{GMV}
    \end{eqnarray*}
for $p/n \to c \in (0,1)$ and $p/m \to \tilde{c} \in (0,\infty)$  as $n,m \to \infty$ by using Lemma \ref{lem:lem3} and Lemma 1.3 in \cite{bodnarokhrinparolya2020}.
    \item The result of part (ii) follows from the proof of Theorem 3.2 of \cite{bodnar2014strong}.
\item We get
\begin{eqnarray*}
\hat{V}_{\hbw_{n;BPS}}&=&\left(\hat{\alpha}_{n;BPS} \hat{\mathbf{w}}_{n;S} + (1- \hat{\alpha}_{n;BPS})\bb\right)^\top \bS_{n+1:n+m} \left(\hat{\alpha}_{n;BPS} \hat{\mathbf{w}}_{n;S} + (1- \hat{\alpha}_{n;BPS})\bb\right)\\
&=&\hat{\alpha}_{n;BPS}^2\hat{\mathbf{w}}_{n;S}^\top\bS_{n+1:n+m} \hat{\mathbf{w}}_{n;S}+(1-\hat{\alpha}_{n;BPS} )^2\bb^\top\bS_{n+1:n+m}\bb\\
&+&2\hat{\alpha}_{n;BPS} (1-\hat{\alpha}_{n;BPS}) \hat{\mathbf{w}}_{n;S}^\top\bS_{n+1:n+m}\bb\\
&=&\hat{\alpha}_{n;BPS}^2\hat{\mathbf{w}}_{n;S}^\top\bS_{n+1:n+m} \hat{\mathbf{w}}_{n;S}+(1-\hat{\alpha}_{n;BPS} )^2\bb^\top\bS_{n+1:n+m}\bb\\
&+&2\hat{\alpha}_{n;BPS} (1-\hat{\alpha}_{n;BPS})
\frac{\bOne^\top\bSigma^{-1/2}\bV_n^{-1} \bV_{n+1:n+m}\bSigma^{1/2}\bb}{\bOne^\top\bS_n^{-1}\bOne}\\
&\stackrel{a.s.}{\rightarrow}&
 \alpha_{BPS}^2\frac{1}{1-c}V_{GMV}+(1-\alpha_{BPS})^2V_{\bb}+2\alpha_{BPS}(1-\alpha_{BPS})V_{GMV}\\
&=&
V_{GMV}+    \alpha_{BPS}^2\frac{c}{1-c}V_{GMV}+(1-\alpha_{BPS})^2(V_{\bb}-V_{GMV})
\end{eqnarray*}
for $p/n \to c \in (0,1)$ and $p/m \to \tilde{c} \in (0,1)$ as $n,m \to \infty$ by applying Lemma \ref{lem:lem3}, Lemma 1.3 of \cite{bodnarokhrinparolya2020}, and the results from parts (i) and (ii).  

\item The result of part (iv) follows from the proofs of parts (i) and (ii).
\end{enumerate}
\end{proof}

\begin{proof}[Proof of Theorem \ref{th:th4}]
The results of the theorem follows from Lemma 1.3 of \cite{bodnarokhrinparolya2020} by noting that 
\[\left|\frac{\bOne^\top\bS_{n+1:m+1}^{-1}\bOne}{\bOne^\top\bSigma^{-1}\bOne}-(1-\tilde{c})^{-1}\right|\stackrel{a.s.}{\rightarrow}0\]
$p/m \to \tilde{c} \in (0,1)$ as $m \to \infty$.
\end{proof}
\bibliography{references}

\end{document}